\documentclass[12pt,aps,tightenlines,notitlepage,superscriptaddress]{revtex4-1}
\usepackage{amsmath}
\usepackage{amssymb}
\usepackage{bm}
\usepackage[mathscr]{eucal}
\usepackage{theorem}
\usepackage{graphicx}
\usepackage{psfrag}
\usepackage{subfigure}
\usepackage{color}
\usepackage{wasysym}
\include{graphix}
\usepackage{framed}
\usepackage{enumitem}%
\usepackage{lipsum}
\usepackage{float}
\usepackage{MnSymbol}
\usepackage{url}

\usepackage{geometry}
\newcommand{\mathd}{\mathrm{d}}
\newcommand{\mathe}{\mathrm{e}}

\newcommand{\hatr}{\hat{r}}
\newcommand{\muhat}{\hat{\mu}}

\newcommand{\myg}{g}
\newcommand{\mygstar}{g_*}
\newcommand{\mygfour}{g^{(4)}}

\newcommand{\Ptot}{P_{\mathrm{tot}}}
\newcommand{\Qtot}{Q_{\mathrm{tot}}}

\newcommand{\Nfunc}{\Psi}

\newtheorem{theorem}{Theorem}[section]
\newtheorem{corollary}[theorem]{Corollary}
\newenvironment{proof}[1][Proof]{\begin{trivlist}
\item[\hskip \labelsep {\bfseries #1}]}{\end{trivlist}}

\newcommand{\myqed}{\nobreak \ifvmode \relax \else
      \ifdim\lastskip<1.5em \hskip-\lastskip
      \hskip1.5em plus0em minus0.5em \fi \nobreak
      \vrule height0.75em width0.5em depth0.25em\fi}

\begin{document}
\title{A differential-equation model shows the proximate cause of the glass ceiling in European academia}
\author{Lennon \'O N\'araigh\footnote{Email address: \texttt{onaraigh@maths.ucd.ie}}}
\address{School of Mathematics and Statistics, University College Dublin, Belfield, Dublin 4, Ireland}

\date{\today}

\begin{abstract}
We introduce a model based on Ordinary Differential Equations to describe how
two mutually exclusive groups progress through a career hierarchy, whether in a single organization, or in an entire economic sector.  The intended application is to gender imbalance at the top of the academic hierarchy in European Universities, however, the model is entirely generic and may be applied in other contexts also. Previous research on gender imbalance in European universities has focused on large-scale
statistical studies. Our model represents a point of departure, as it is deterministic
(i.e. based on Ordinary Differential Equations).   The model requires a precise definition of the progression rates for the different groups through the hierarchy; these are key parameters governing the dynamics of career progression.
The progression rate for each group can be decomposed into a product: the proportion of group members at a low level in the hierarchy who compete for promotion to the next level a given year, multiplied by the in-competition success rate for the group in question.  Either of these two parameters can differ across the groups under consideration; this introduces a group asymmetry into the organization's composition.  We introduce a glass-ceiling index to summarize this asymmetry succinctly.  Using case studies from the literature, we  demonstrate how the mathematical framework can  pinpoint the proximate cause of the glass ceiling in European academia.
\end{abstract}

\maketitle

\section*{Introduction}

In Europe, women and men make up roughly equal numbers at the entry level of the academic career hierarchy.  However, European academia has a glass ceiling: the proportion of women at the top of the academic hierarchy is less than the proportion of women in the university system as a whole~\cite{she2015}.  To understand the root cause of this effect, a range of different factors has been investigated, including implicit bias~\cite{gvozdanovic2018implicit}, stereotypes~\cite{van2015gender}, the different emphasis placed on family life by women and men~\cite{santiago2012personal}, and the different institutional responsibilities  undertaken by women and men~\cite{poole1997international}.  And yet, conclusions on the role of these factors are still unclear 
(e.g.~\cite{ceci2015women,volker2015no}).  Therefore, in order to make progress in understanding the glass ceiling in European academia,  we propose instead to look at the proximate cause of the glass ceiling.  By this focus, it is proposed to bring some clarity to the above debate, with a view to formulating more precise research questions regarding the underlying cause of the glass ceiling in European academia.  

The present work is focused on Europe, this is justified by the availability of extensive official statistics and cross-country comparisons~\cite{she2015}, as well as large-scale studies of the promotion systems in Spain, France, and Italy~\citep{zinovyeva2010does,bosquet2014gender,de2015gender,de2015gender_b}.  The latter  are based on regression analyses whereby the probability of promotion is linked to gender, as well as other factors.  This is an important first step in understanding the proximate cause of the glass ceiling in European universities.  We emphasize that the mathematical model developed here is generic, and will have applicability beyond Europe as well.

This work is different but complementary to the above statistical-modelling approach.  We formulate an abstract model of an organizational hierarchy, broken into two mutually exclusive groups, the $P$-group and the $Q$-group.    In the present context, the $P$-group refers to men and the $Q$-group to women, however, the model is completely general, and can be applied equally well to other pairs of mutually exclusive groups which together make up the organization.  At the same time, and for simplicity, the organizational hierarchy is assumed to be binary, with an entry level and a managerial level.  Members of both groups progress from the entry level to the managerial level, but the progression rates for the $P$-group and the $Q$-group are not necessarily the same in the model.  We also extend the the model to hierarchies with more than two levels.  We emphasize  that the model applies equally to a single organization, or to an entire economic sector; we use the term `organization' throughout this article for definiteness.

\section{Formulation}

We set up a mathematical model of an idealized organization made up of $N$ employees.  We model how $N$ varies with time $t$, starting from a given set of initial conditions at $t=0$.  The employees can be categorized in two well-defined, non-overlapping groups, the $P$-group and the $Q$-group, with total populations $\Ptot(t)$ and $\Qtot(t)$, respectively, such that $\Ptot(t)+\Qtot(t)=N(t)$.
The organizational hierarchy is assumed to be simple, with two levels -- an entry level (labelled with a subscript 1) and a managerial level (labelled with a subscript 2).  As such, the $P$-group can be decomposed as $P_1(t)+P_2(t)=\Ptot(t)$, where $P_1$ denotes the number of members of the $P$-group at the entry level; a similar decomposition applies to the $Q$-group also.  Hence, $P_1(t)+P_2(t)+Q_1(t)+Q_2(t)=N(t)$.  A more complicated four-level model is also developed (Section~\ref{sec:application}).  The model is further assumed to have the following characteristics:
\begin{enumerate}[noitemsep]
\item \label{ass_time} Time is measured in years.
\item \label{ass_homo} Apart from the binarization of the population into the specified groups and career levels, the population is otherwise homogeneous.
\item \label{ass_N} The total organizational headcount grows according to $\mathd N/\mathd t=\lambda N$,
where $\lambda$ is a constant with  dimensions of $[\text{Number of indivudals}]^{-1}[\text{Year}]^{-1}$.
\item \label{ass_resign} There are no resignations, deaths in service, redundancies, or dismissals -- employees leave the organization only through retirements.
\item \label{ass_equal} The organization recruits members of both groups at equal rates.  Recruitment is  only at the entry level; access to the managerial level is by promotion only.    Once at the managerial level, employees cannot return to the entry level -- there is no `demotion' of managers.
\item \label{ass_retire} Employees of the $P$- and $Q$-groups retire at equal rates; employees at the different levels in the hierarchy retire at different rates.  
\item \label{ass_rectirex} There is an overall `crude' retirement rate set by the average length of service.
\item \label{ass_fixed} The total number of employees at the managerial level is constrained, with $(P_2+Q_2)/N=\varphi$, where $0<\varphi<1$ is a constant.  Correspondingly, $(P_1+Q_1)/N=1-\varphi$. 
\end{enumerate}
Further discussion about the appropriateness of these assumptions to a university setting is given in the {\textbf{Appendix}}.


\subsection*{Model equations}

Based on the above assumptions, we introduce the following pair of ODEs to describe the composition of the $P$-group:
\begin{subequations}
\begin{eqnarray}
\frac{\mathd P_1}{\mathd t}&=&s - r_1 P_1 - \mu P_1,\\
\frac{\mathd P_2}{\mathd t}&=& \mu P_1-r_2 P_2.
\end{eqnarray}%
\label{eq:model1}%
\end{subequations}%
Here, $s$ is the source function governing the rate at which individuals of the $P$-group are recruited into the organization; this quantity has dimensions of 
$[\text{Number of individuals}][\text{Year}]^{-1}$.  The source function $s$ depends on time.  The other coefficients in~\eqref{eq:model1} have the interpretation of rate coefficients, possibly time-dependent.
The rate coefficients $r_1$ and $\mu$ have physical units of $[\text{Percentage}][\text{Year}]^{-1}$; $r_1$ can be interpreted as the proportion of a $P_1$-individuals who retire per year (`retirement rate'), and $\mu$ can be interpreted as the proportion of all $P_1$-individuals who are promoted to the managerial level, per year (`progression rate').  Similarly, $r_2$ is the proportion of $P_2$-individuals who retire per year. 
Finally, it should be noted that~\eqref{eq:model1} is valid for $t>0$; at $t=0$, initial conditions apply, e.g. $P_1(t=0)=P_{10},\qquad P_2(t=0)=P_{20}$,  where $P_{10}$ and $P_{20}$ are constants.

The equations  for the $Q$-group are very similar to those already written down for the $P$-group:
\begin{subequations}
\begin{eqnarray}
\frac{\mathd Q_1}{\mathd t}&=&s - r_1 Q_1 - \mu' Q_1,\\
\frac{\mathd Q_2}{\mathd t}&=& \mu' Q_1-r_2 Q_2.
\end{eqnarray}%
\label{eq:model2}%
\end{subequations}%
The source function and the retirement rates are the same in~\eqref{eq:model1} and~\eqref{eq:model2}, as per Assumptions~\eqref{ass_equal}--\eqref{ass_retire}.  Equality of the source functions for the different groups is justified by the official statistics~\cite{she2015}.  The progression rate in~\eqref{eq:model2} is $\mu'$.  
The model therefore allows for asymmetric progression rates $\mu\neq \mu'$.

\eqref{eq:model2} is valid for $t>0$; at $t=0$, initial conditions again apply at $t=0$, e.g. $Q_1(t=0)=Q_{10},\qquad Q_2(t=0)=Q_{20}$. The initial conditions for the $P$- and $Q$-groups are not all independent, indeed, we must have $P_{10}+P_{20}+Q_{10}+Q_{20}=N_0$,
where $N_0$ is the total headcount at $t=0$; moreover, we must have $P_{10}+Q_{10}=(1-\varphi)N_0$ and $P_{20}+Q_{20}=\varphi N_0$.
The motivation for writing down Equations~\eqref{eq:model1}--\eqref{eq:model2} in their present form is that they represent a `conservation-of-people' or `balance' principle, whereby entry into the organization, and exit therefrom, are governed by well-defined mechanisms.  In particular, individuals enter the organization via recruitment and leave via retirement.  Individuals may move from one level in the hierarchy to the next -- the loss of one person from the entry level represents a gain at the managerial level.  As such~\eqref{eq:model1}--\eqref{eq:model2} represent the simplest possible set of equations that can be written down that describe this balance principle.

\section{Theoretical Analysis}

The basic model~\eqref{eq:model1}--\eqref{eq:model2} contains a plethora of parameters.  However, not all of these are independent.  We have proved a number of theorems whose purpose is to reduce the number of independent parameters down to a minimum.  We summarize these results here.  Detailed discussion is provided in the {\textbf{Appendix}}.
\begin{theorem}
The source term is not arbitrary; it is given by
\begin{equation}
s(t)=\tfrac{1}{2}\left(\lambda+\hatr\right)N_0\mathe^{\lambda t},
\label{eq:source}
\end{equation}
where $\hatr$ is the crude retirement rate $\hatr=r_1(1-\varphi)+r_2\varphi$.
\end{theorem}
The crude retirement rate $\hatr$ is known -- it is simply $r=1/T$, where $T$ is the average length of service.  Similarly, the retirement rate $r_2$ is known -- if $T_*$ is the average time between recruitment and promotion, then $r_2=(T-T_*)^{-1}$.  As such, we have the following theorem:
\begin{theorem}
The retirement rate $r_1$ is given by
\begin{equation}
r_1=(\hatr - r_2\varphi)/(1-\varphi).
\label{eq:r1}
\end{equation}
\end{theorem}
We next define a crude progression rate $\muhat$, such that $\muhat (P_1+Q_1)= \mu P_1+ \mu' Q_1$.
We have:
\begin{theorem}
The crude progression rate $\muhat$ is given by
\begin{equation}
\muhat=(r_2+\lambda)\varphi/(1-\varphi).
\label{eq:r1}
\end{equation}
\end{theorem}
With the value of $\muhat$ prescribed, it is possible to write the independent progression rates $\mu$ and $\mu'$ in a more succinct form.  As such, we write $\mu'=k\mu$, where $k$ is a non-negative constant.  If $k<1$ there is a preference for the $P$-group in the promotion system; otherwise, if $k>1$ there is a preference for the $Q$-group.  Given the definition of $\muhat$, we have $\mu=\muhat (P_1+Q_1)/(P_1+kQ_1)$, or $\mu=\muhat \Nfunc(P_1,Q_1)$, where 
 $\Nfunc=(P_1+Q_1)/(P_1+kQ_1)$ is homogeneous in each of its variables, i.e. $\Nfunc(xP_1,xQ_1)=\Nfunc(P_1,Q_1)$, for all $x\neq 0$.

In an analysis of the structure of the organization in terms of the $P$- and $Q$-groups, what is of interest is not the headcounts $P_1,\dots, Q_2$ but rather the proportion of individuals at a given career level.  As such, we introduce the scaled variables $p_1=P_1/N$, $p_2=P_2/N$, $q_1=Q_1/N$, and $q_2=Q_2/N$.
We have the following theorem:
\begin{theorem}
Given~\eqref{eq:model1}--\eqref{eq:model2}, the scaled variables  satisfy the following ODEs:
\begin{subequations}
\begin{eqnarray}
\frac{\mathd p_1}{\mathd t}&=&s_0-(r_1+\lambda) p_1 - \muhat \Nfunc(p_1,q_1)p_1,\\
\frac{\mathd p_2}{\mathd t}&=&\muhat \Nfunc(p_1,q_1)p_1-(r_2+\lambda)p_2,\\
\frac{\mathd q_1}{\mathd t}&=&s_0-(r_1+\lambda) q_1 - \muhat k \Nfunc(p_1,q_1)q_1,\\
\frac{\mathd q_2}{\mathd t}&=&\muhat k \Nfunc(p_1,q_1)q_1-(r_2+\lambda)q_2,
\end{eqnarray}%
\label{eq:syst2}%
\end{subequations}%
where $s_0=s(t)/N(t)=(\lambda+\hatr)/2$.
\end{theorem}
The proof follows by direct computation; the homogeneity of the 
 $\Nfunc$ is a key part of the computation.
In the remainder of the paper we work with the scaled model in~\eqref{eq:syst2}. 

\section{Steady-state solution}
\label{sec:steady}

We examine steady-state solutions of~\eqref{eq:syst2} obtained by setting the time derivatives on the left-hand side equal to zero.   The results give further insight into the structure of the model as well as giving the motivation for introducing the glass-ceiling index.
As such, we obtain
\begin{subequations}
\begin{eqnarray}
s_0=(r_1+\lambda)p_1+\muhat \Nfunc p_1,&\phantom{a}&\muhat \Nfunc p_1=(r_2+\lambda)p_2,\\
s_0=(r_1+\lambda)q_1+k\muhat \Nfunc q_1,&\phantom{a}&\muhat k\Nfunc q_1=(r_2+\lambda)q_2.
\end{eqnarray}
\end{subequations}
These are algebraic equations which can be solved to give explicit solutions.
%

The case $k=0$ is special, and corresponds to the steady state $q_2=0$.  The full solution in this case is discussed in the {\textbf{Appendix}}.  The solution comes with the requirement
\begin{equation}
\varphi\leq (\hatr+\lambda-s_0)/(r_2+\lambda).
\label{eq:varphi_cst}
\end{equation}
As $k=0$ is an extreme case,~\eqref{eq:varphi_cst} can be thought of as a sufficient condition such that $p_1\geq 0$ at the steady state; indeed, this can be assumed to be a general condition to avoid a population crash where $p_1\rightarrow 0$ in finite time.  As such, in the remainder of this work, we assume that~\eqref{eq:varphi_cst} holds. 

A second special case is $k=1$, which corresponds to symmetric $p$- and $q$- populations.  In this case, it is readily seen that $p_2=q_2=\varphi/2$ and $p_1=q_2=(1-\varphi)/2$ in the steady state.  Otherwise, a general solution pertains; in the {\textbf{Appendix}} is is shown by straightforward calculation that the general solution is parametrized as follows (we use an asterisk to denote the steady state):
%
\begin{equation*}
(p_{1*},p_{2*},q_{1*},q_{2*})=\\
\left(\frac{1-\varphi}{1+(x/k)},
\frac{\varphi}{1+x},
\frac{x}{k}\frac{1-\varphi}{1+(x/k)},
\frac{x\varphi}{1+x}\right),
\label{eq:pstar}
\end{equation*}
where $x$ is the solution (positive branch) of a quadratic equation:
\begin{equation}
x^2+(k-1)x\left[(\varphi/s_0)(r_2+\lambda)-1\right]-k=0,
\label{eq:x4}
\end{equation}

We introduce the \textit{glass-ceiling index} $g(t)$ to display these results graphically:
\begin{multline}
\myg(t)=\frac{\text{Proportion of organization made up by }Q\text{-group}}{\text{Proportion of managerial level made up by }Q\text{-group}}\\
=\frac{ (Q_1+Q_2)/N}{ Q_2/(Q_2+P_2)}
=\varphi\left(1+\frac{q_1}{q_2}\right).
\label{eq:glassceiling}
\end{multline}
Correspondingly, we introduce $\mygstar=\lim_{t\rightarrow \infty}\myg(t)$.  Hence, from Equation~\eqref{eq:pstar}, we have
\begin{equation}
\mygstar=\varphi\left(1+\frac{1-\varphi}{\varphi}\frac{1+x}{k+x}\right).
\label{eq:glassceiling1}
\end{equation}
A sample curve $\mygstar(k)$ is shown in Figure~\ref{fig:glass_ceiling}.  The value $\mygstar=1$ corresponds to $k=1$, hence complete symmetry between the $p$- and $q$-groups.  
\begin{figure}[htb]
	\centering
		\includegraphics[width=0.8\linewidth]{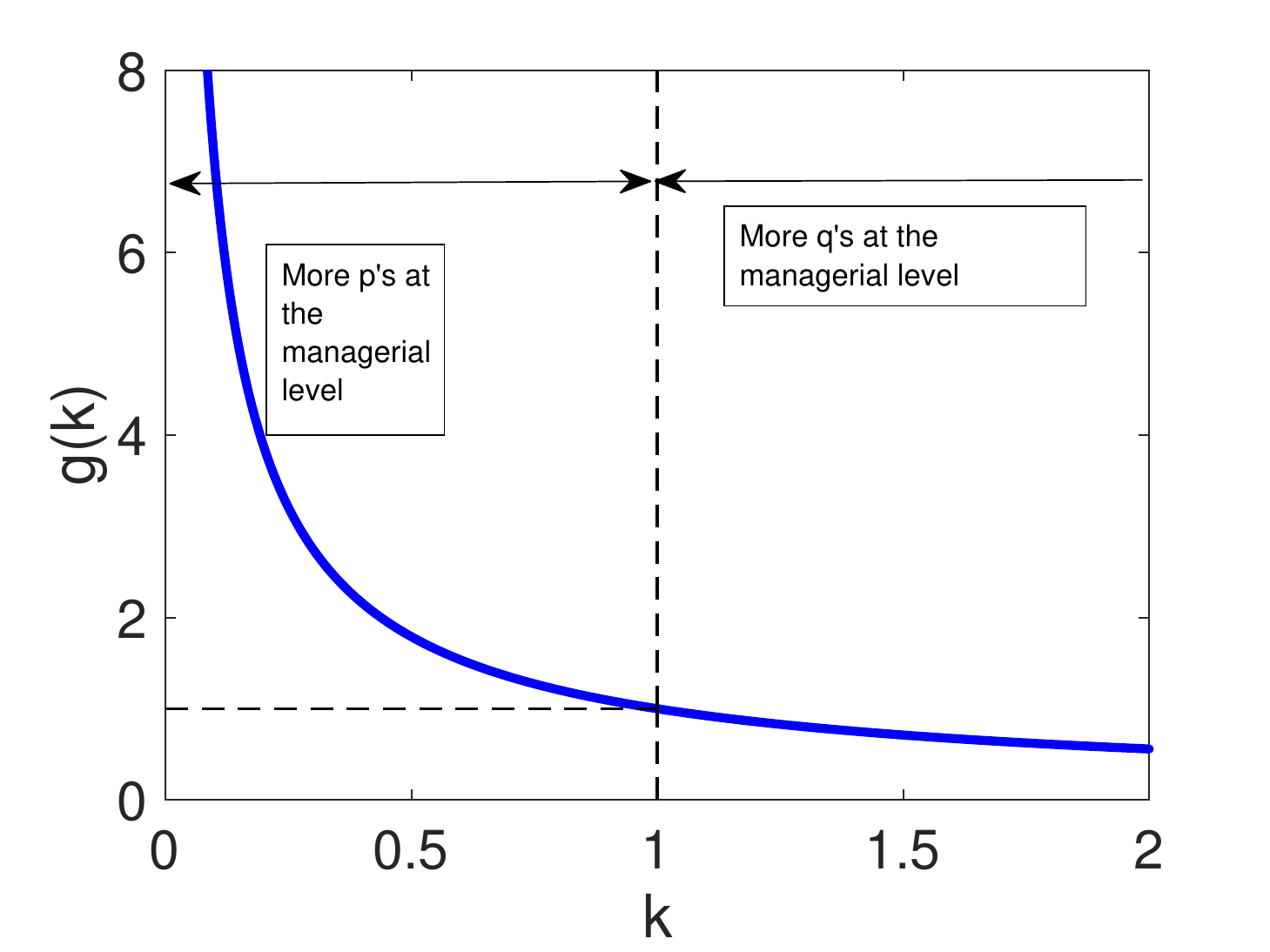}
		\caption{Glass-ceiling index $\mygstar$ as a function of the asymmetry parameter $k$, as given by~\eqref{eq:glassceiling1}.  
		Model parameters -- for illustration purposes only: $\lambda=0$ (steady-state headcount), $\hatr=(1/35)\,\text{[Years]}^{-1}$, $r_2=(1/15)\,\text{[Years]}^{-1}$,
		$\varphi=0.245$}
	\label{fig:glass_ceiling}
\end{figure}

\section{Dynamics}

\begin{figure}[htb]
\centering
\includegraphics[width=0.95\linewidth]{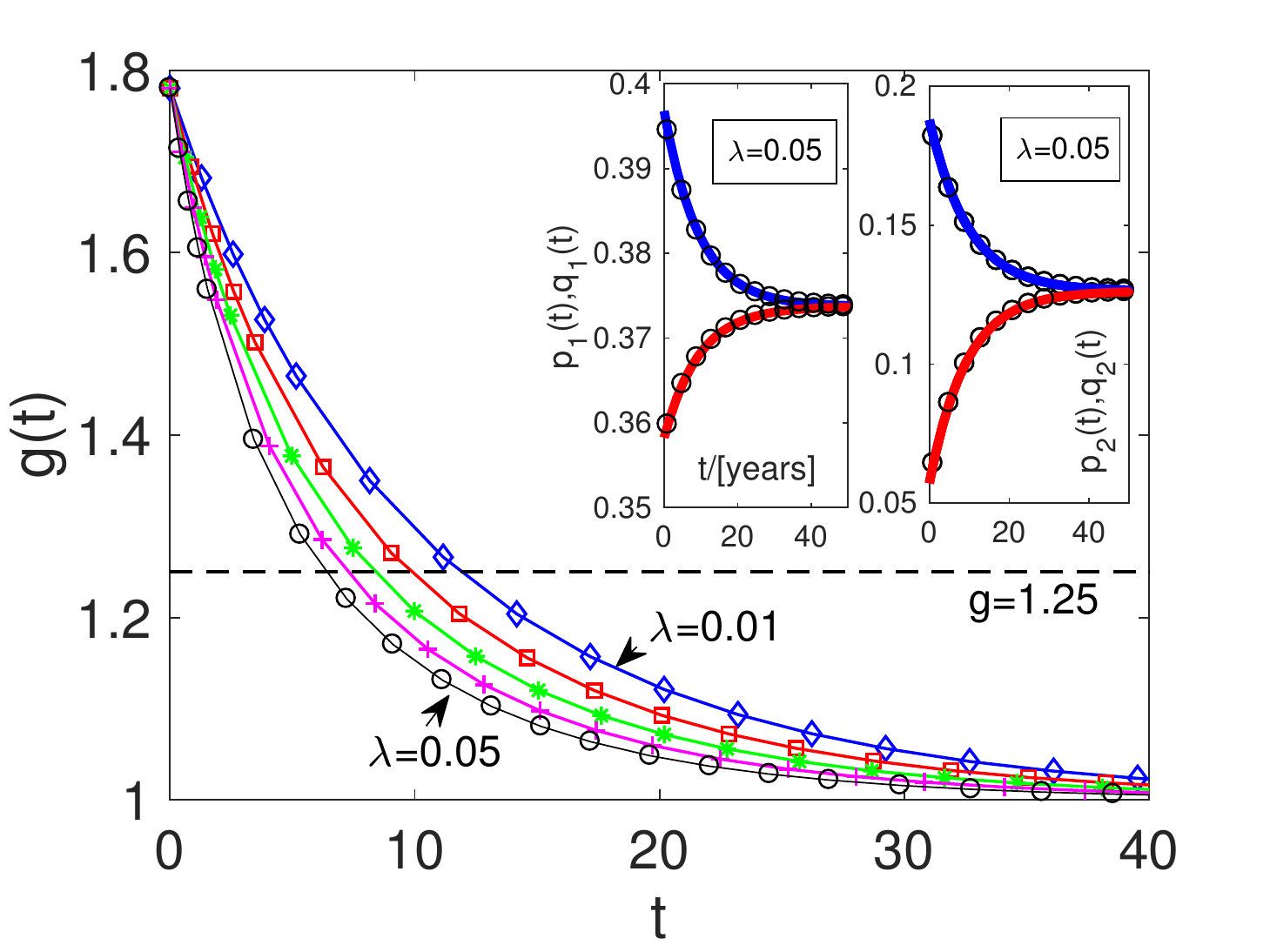}\\
\caption{The glass-ceiling index $\myg(t)$ as a function of time for various values of the headcount growth rate $\lambda$, starting with $\lambda=0.01$ and moving monotonically through the values $\lambda=0.02,0.03,0.04$, to $\lambda=0.05$.  The initial conditions can be read from the graph in the inset.  Other parameters common to all curves: $\hatr=1/35\,\text{[Years]}^{-1}$, $r_2=1/15\,\text{[Years]}^{-1}$.   }
\label{fig:test2}
\end{figure}

When $k=1$ the $p$- and $q$-populations are symmetric and the populations relax over time to the steady state $p_2=q_2=\varphi/2$ and $p_1=q_1=(1-\varphi)/2$.  In this scenario, we have $\Psi=1$, and the ODEs for the $p$- and $q$-groups become decoupled;  indeed, the model ODEs reduce to a linear set of ODEs which can be solved via the standard integrating-factor technique.  The details of these solutions are provided in the {\textbf{Appendix}}; what matters is the appearance of the exponential factors $\mathe^{-(r_1+\muhat+\lambda)t}$ and $\mathe^{-(r_2+\lambda)t}$ in the solutions, which multiply the initial conditions.  As such, the effect of the initial conditions is attenuated over time, the attenuation is governed by the `decay times' $\tau_1=(r_1+\muhat+\lambda)^{-1}$ and $\tau_2=(r_2+\lambda)^{-1}$.

In universities, positive headcount growth  is generally seen as positive (in other contexts it might be a sign of organizational inefficiency).  Hence, we assume $\lambda\geq 0$ in the remainder of this article.  Attenuation is only as fast as the longest timescale of $\tau_1$ and $\tau_2$.  As such, for $\lambda\geq 0$, we identify an overall attenuation rate $\tau_*=\max(\tau_1,\tau_2)=(r_2+\lambda)^{-1}$.  From this expression, 
 it is clear that both the retirement rate $r_2$ and a headcount growth rate $\lambda\geq 0$ act together to attenuate the initial conditions and to hasten the onset of the steady state.  As such, the onset of the steady state $p_2/p_1=q_2/q_1=\varphi/(1-\varphi)$ can be hastened by increasing either $r_2$ or $\lambda$.

We use numerical solutions of \eqref{eq:syst2} to demonstrate the effect of $\lambda$ on the glass-ceiling index -- this is shown in Figure~\ref{fig:test2}.    An increase in the headcount growth rate $\lambda$ accelerates the convergence of the glass-ceiling index to its long-term steady value.  For instance, if an organization deems it desirable to have 40\% female representation at the management level and 50\% female representation in the organization as a whole, this corresponds to $g=1.25$.  For the parameters in Figure~\ref{fig:test2}, this can be achieved in 12 years for $\lambda=0.01\,\text{[Years]}^{-1}$, and in only 6 years for $\lambda=0.05\,\text{[Years]}^{-1}$. 
%

\section{Decomposition of the progression rate}

We decompose the progression rates, thereby obtaining a more fundamental understanding of the asymmetry $\mu\neq \mu'$.  The decomposition for the $P$-group is given by
\begin{multline}
\mu=\left[\begin{array}{c}\text{Number of }P\text{-individuals moving to the managerial level,}\\
\text{as a proportion of all }P\text{-individuals at the entry level}\end{array}\right]\Big/\left[\text{Year}\right]\\
=\left[\begin{array}{c}\text{Number of }P\text{-individuals under consideration for promotion,}\\
\text{as a proportion of all }P\text{-individuals at the entry level}\end{array}\right]\\
\times \left[\begin{array}{c}\text{Success rate of the} 
P\text{-individuals in the promotion system}\end{array}\right]
\Big/\left[\text{Year}\right]
=\nu\times \sigma,
\end{multline}
where $\nu$ is the first factor in the product $\mu=\nu\times\sigma$ and $\sigma$ is the second factor -- as such $\nu$ has dimensions of 
$\left[\text{Year}\right]^{-1}$ and $\sigma$ is a a pure percentage, with $\sigma \in [0,1]\times 100\%$.  We similarly write $\mu'=\nu'\times\sigma'$ (in an obvious notation).  


Accordingly, we classify the promotion system according to the parameters $\nu$, $\nu'$, $\sigma$, and $\sigma'$ as follows:
\begin{itemize}[noitemsep]
\item Supply-side bias: $\sigma=\sigma'$, $\nu\neq \nu'$.
\item In-competition bias: $\nu=\nu'$, $\sigma\neq \sigma'$.
\item Multiple biases: $\nu\neq \nu'$ and $\sigma\neq \sigma'$.
\item Symmetry: $\nu=\nu'$, $\sigma=\sigma'$.
\end{itemize}
\subsection*{The Cascade Model}

In certain contexts~\citep{cascade2018}, the so-called cascade model of promotions has been implemented to bring about equality between the $P$-group and the $Q$-group at the managerial level.  The cascade model stipulates that the proportion of $P$s and 
$Q$s to be recruited or promoted to a certain level is based on the proportion of each at the career level directly below.   This requirement is equivalent to enforcing $k=1$, in the context of our model:
\begin{theorem}
\label{thm:cascade}
The cascade model requires that $k=1$.
\end{theorem}
The proof is straightforward and is provided in the {\textbf{Appendix}}.
From the classification of biases, it can be seen that $k=1$ does not exclude bias: the cascade model is free from bias provided only supply-side and in-competition biases are both eliminated.  As such, we have the following corollary to Theorem~\ref{thm:cascade}
\begin{corollary}
If the promotion system has a supply-side bias, then implementation of the cascade model requires the introduction of a compensatory in-competition bias.
\end{corollary}

\section{Application to Gender Balance in European Universities}
\label{sec:application}

We apply our model to academic staff in European universities, with  $Q=\text{women}$ and $P=\text{men}$.
For the purposes comparison between different Member States of the European Union, the  academic career hierarchy is standardized across the European Union and is broken into four levels~\citep{she2015}.  Using common terminology, Level D corresponds to the entry level (PhD / Postdoc), Level C corresponds to Assistant Professor, Level B to Associate Professor, and Level A to Full Professor.   This makes comparison between the European Union data and our own simplified two-level model difficult.  We therefore extend our two-level model to allow for four levels.  We do not solve the resulting equations; instead, it suffices to identify the limiting step for promotions to the highest level and thereby to compare with the official data.  The official data used for these purposes are the summary statistics published by 
the European Commission (the `She Figures')~\citep{she2015}. 

\subsection*{Four-level model}

We introduce a four-equation model, by direct generalization from the two-level~\eqref{eq:model1}--\eqref{eq:model2}.  For the $P$-group, we have:
\begin{subequations}
\begin{eqnarray}
\frac{\mathd P_D}{\mathd t}&=&S_P(t)        -r_D P_D - \mu_D P_D,\\
\frac{\mathd P_C}{\mathd t}&=&\mu_D P_D     -r_C P_C - \mu_C P_C,\\
\frac{\mathd P_B}{\mathd t}&=&\mu_C P_C     -r_B P_B - \mu_B P_B,\\
\frac{\mathd P_A}{\mathd t}&=&\mu_B P_B     -r_A P_A.
\end{eqnarray}
Similarly, for the $Q$-group, we have
\begin{eqnarray}
\frac{\mathd Q_D}{\mathd t}&=&S_Q(t)        -r_D Q_D - \mu_D' Q_D,\\
\frac{\mathd Q_C}{\mathd t}&=&\mu_D'Q_D     -r_C Q_C - \mu_C' Q_C,\\
\frac{\mathd Q_B}{\mathd t}&=&\mu_C' Q_C     -r_B Q_B - \mu_B' Q_B,\\
\frac{\mathd Q_A}{\mathd t}&=&\mu_B' Q_B     -r_A Q_A.
\end{eqnarray}%
\label{eq:fourlevel}%
\end{subequations}%
Here, the rate coefficients $r_A,\dots, r_D$ characterize retirements, these are assumed to be the same for the $P$-group and the $Q$-group.  The coefficients $\mu_A,\dots,\mu_D,\mu'_A,\dots,\mu'_D$ are progression rates, which may be different for the $P$-group and the $Q$-group.  Finally, $S_P(t)$ and $S_Q(t)$ are source functions (possibly distinct); these are chosen such that the total headcount 
\[
N=\sum_{i=A,B,C,D}P_i+\sum_{i=A,B,C,D}Q_i
\]
grows at a set rate, $\mathd N/\mathd t=\lambda N$.
 
In this context, the glass-ceiling index defined in the `She Figures',
\begin{equation}
\mygfour(t)=\frac{Q_A+Q_B+Q_C}{Q_A+Q_B+Q_C+P_A+P_B+P_C}\big/ \frac{Q_A}{Q_A+P_A}
\label{eq:gfour}
\end{equation}
can be applied directly to the four-level model~\eqref{eq:fourlevel}. 
From the definition~\eqref{eq:gfour}, the limiting factor which prevents $\mygfour$ from being equal to one is the progression from level B to level A.  As such, in the remainder of this section, we study the progression rates $\mu_B$ and $\mu'_B$.  We decompose these as 
$\mu_B=\nu\sigma$ and  $\mu'_B=\nu'\sigma'$,
where in this new context $\nu$ denotes the proportion of men at level B entering the competition for promotion to level A in a typical year, and $\sigma$ denotes the success rate of men in the competition for promotion from level B to level A.  The same interpretation holds for $\nu'$ and $\sigma'$ for women.
As such we explore whether supply-side bias, in-competition bias, or a combination of both, is applicable in selected European countries.

We first of all provide some context for our study, using official data of the European Commission for the year 2015~\citep{she2015}.  The average value was $\mygfour=1.75$, averaged over all Member States.  This corresponds to more men than women at Level A.  The glass-ceiling index for selected countries is shown in Table~\ref{tab:she2}.
\begin{table}
\centering
\caption{Academic hierarchy by gender, selected countries.  Taken from the `She Figures 2015' (Table 6.1 therein).}
%
\begin{tabular}{|c|c| c| c |c|c| }
		\hline
		Country                  &  D (F/M) &   C (F/M)    &  B (F/M)  &   A (F/M)  &  $\mygfour$         \\
		\hline
			France           	      & 41/59 &  30/70  & 40/60 & 19/81 & 1.72\\
		  Italy                   & 50/50 &  45/55  & 35/65 & 21/79 & 1.73\\
			Spain                   & 51/49 &  49/51  & 40/60 & 21/79 & 1.80\\
			EU-28                   & 47/53 &  45/55  & 37/63 & 21/79 & 1.75\\
			\hline
		\end{tabular}
\label{tab:she2}
\end{table}
To understand the reasons behind the  different values of $\mygfour$ in Table~\ref{tab:she2}, we discuss our model in the context of econometric literature on gender balance in the hierarchy in European Universities.
We focus our review on France, Italy, and Spain, which currently (2019) organize national central competitions to determine promotion to the highest academic grade.    

\subsection*{Spain}

Zinovyeva and Bagues~\cite{zinovyeva2010does} have looked at the Spanish academic promotion system  in the years 2002-2006.  In these years, the system was based a national examination (\textit{Habilitaci\'on}).  Expert evaluation committees  were convened, by random selection from a national pool of professors.  
%
The authors  used a regression model is used to determine how the probability $p$ of success in the competition depends on the gender of applicants, as well as other applicant attributes (age, academic productivity, etc.).  By random selection, some evaluation committees  have an all-men composition -- this facilitated a `natural experiment' whereby the effect of the committee composition on promotion prospects can be studied.
For academics in the competition,  it was found that $p$  depends on gender -- men had on average a higher probability of success.  The difference is $p=10.2\%$ for men and $p=8.9\%$ for women.  This was found to be statistically significant.  The difference vanished  for  mixed-gender panels.

Using the above findings, and the summary statistics included in the study, we have identified values which correspond to $\sigma$, $\nu$, etc. in our model~\eqref{eq:fourlevel}.  We use the given data that in the period reported, men made up $65\%$ of the headcount at the Associate Professor level, and women made up $35\%$ -- these combined make up the pool of applicants for the Full-Professor level.  At the same time, there were 6,037 individuals who applied for promotion to the top level.  However, many applicants applied more than once -- the average number of applications per candidate in the given time period was just over two.  As such, there were 9,480 applications by men candidates, and 3,744 applications by women candidates, corresponding to a total of 13,224 applications.  From these numbers, the ratio $\nu/\nu'$ may be expressed as $\nu/\nu'=(9480/0.65)/(3744/0.35)\approx 1.36$.
Equally,  $\sigma$ and $\sigma'$ can be identified with the average probability of success for men and women candidates respectively, hence $\sigma/\sigma'=10.2\%/8.9\%\approx 1.14$.
Using the difference operator $\Delta\Phi=\Phi_{P}-\Phi_{Q}$, where $\Phi$ is any one of $\mu$, $\nu$, and $\sigma$, we obtain the following relation via the standard expression for the difference of a product:
\begin{equation}
\frac{\Delta (\sigma\nu)}{\sigma \nu}=\frac{\Delta \nu}{\nu}+\frac{\Delta \sigma}{\sigma}.
\end{equation}
As such, we can identify $\Delta \nu/\nu=(1.36-1)/1.36=0.26$ and $\Delta \sigma/\sigma=(1.14-1)/1.14=0.12$.  Hence, the effect of $\nu\neq \nu'$ contributes twice as strongly as the effect of $\sigma\neq \sigma'$ in determining the asymmetry in the progression rate between men and women.  In other words, the system is asymmetric between men and women mostly because of the supply-side effect, but in-competition bias plays a role also.

\subsection*{France}

Bosquet et al.~\cite{bosquet2014gender} have examined the French academic promotion system  in Economics between 1991-2008, where promotion (in all subjects) is based on a national competition.  
The data consist of academics who applied for promotion, and those who did not.
The authors were therefore able to distinguish between these two groups, and introduced a probability $p(S)$ of success for a candidate, conditional on his/her having applied for promotion.  
The authors constructed a regression model for $p(S)$.  The results showed a differences in $p(S)$ for men and women, but they were  not  statistically significant.  The authors were further able to demonstrate that the main reason for the gender asymmetry at the top of the hierarchy is in the difference in the proportion of men and women who enter the promotion competitions, i.e. a supply-side effect.  The regression analysis can therefore be used to compute the parameters of our own model, $\sigma$, $\nu$, etc.  In summary, the data fro France are imply $\sigma=\sigma'$ but $\nu\neq \nu'$ (full details provided in the {\textbf{Appendix}}).

\subsection*{Italy -- Local Competitions}

De Paola and Scoppa~\cite{de2015gender} have examined the Italian academic promotion system for Associate Professor (Level B) and Full Professor (Level A) positions in the years 2008-2011, when the system was based on institutional-level (local) competitions.  The data for the entire country were collated by the Ministry for Education and made available for research purposes.  The reference focuses on competitions  in Economics and Chemistry.  Competitions for promotions to Associate and Full Professor are grouped together. 
The authors used a regression model to determine how the probability $p$ of success in the competitions depends on gender, as well as other factors.  A statistically significant dependence on gender was found, with $p=8.9\%$ for women and with $p=13\%$ for men.  However, the difference almost vanished in cases where candidates were assessed by mixed-gender panels.
 Based on the given values of $p$, the number of candidates entering the competition, and the composition by gender of the pool of potential candidates for the Full Professor positions, we have computed $\Delta \nu/\nu=0.06$ and $\Delta \sigma/\sigma=0.40$ (full details in the {\textbf{Appendix}}).  Hence, the asymmetry in the progression rates for progression to Level A  is due almost entirely to in-competition bias.  

\subsection*{Italy -- National Competitions}

A reformed, national-level competition (\textit{Abilitazione Scientifica Nazionale}, ASN) for promotion to both Associate Professor and Full Professor has been introduced in Italy, and the results for the years 2012-2014 were studied in a follow-on paper by de Paola et al.~\cite{de2015gender_b}.  Data for the entire set of Italian academics were examined (i.e. data corresponding to academics who do not enter the competitions, as well as those that do).  The authors were therefore able to examine the probability that an individual academic will enter the competition, as well as the probability that an individual academic will succeed in the competition, conditional on having entered.  The authors found that overall, a female academic has a statistically significant lower probability of entering the competition.
Interestingly, the main effect documented in the first study~\citep{de2015gender_b} (different male/female success rates) vanished for the ASN -- the probability that an individual academic will succeed in the competition, conditional on having entered, is the same for men and women (no statistically significant difference).  Hence, in the reformed competition, the asymmetry in the progression rates for progression to Level A  is again due to a supply-side effect.

\section{Discussion and Conclusions}

Summarizing, we have introduced a deterministic model  to describe how  the structure of an organization's hierarchy changes over time.
The model categorizes staff into two groups ($P$-group and $Q$-group) -- and makes a prediction for the long-time balance between the $P$-group and the $Q$-group at the top of the hierarchy. 
The model is broadly applicable in a public-sector context with security of tenure, and most applicable in a university setting where headcount growth is a desirable organizational aim.   The model is generic, however we have applied it in a context where $P=\text{men}$ and $Q=\text{women}$.
%
%
%

We have defined precisely the progression rates which control the long-term steady state of the model.  We have decomposed the progression rates into their component parts -- this pinpoints 
the proximate cause of any gender imbalance in the organizational hierarchy -- supply-side bias and in-competition bias.

We have extended the model to a four-level system and applied it to European University systems.  Using  data from large-scale studies of three  countries, we have estimated the relevant  model parameters (in particular, the components of the progression rates for progression to the highest academic level) for France, Italy, and Spain.  
The analysis reveals that supply-side effects play the main role in the gender imbalance at the top of the academic hierarchy in these countries.

It will be of interest in future work to gather data (progression rates, retirement rates, etc.) for a range of countries, and thereby to make predictions concerning the long-term composition of the academic hierarchy throughout the European Union (and beyond).    As part of any such study, the question of retirement rates should be revisited -- in the countries studied herein these are the same for men and women, although this should be checked carefully for each country under consideration.    Indeed, a quick sensitivity analysis of the underlying differential equations in our model indicates that differential retirement rates will also have a strong effect on the steady-state model solutions.  

Furthermore, it will be of interest to revisit Assumption~\ref{ass_homo}, namely the homogeneity of the populations.  In studies of gender equality, paradoxes of amalgamation often arise, whereby an amalgamated success rate tells a different story from success rates broken down by discipline~\cite{bickel1975sex,walton2016simpson}.  Therefore, in future, the modelling approach developed herein may be applied on a discipline-by-discipline level, eventually building from the bottom up to a comprehensive amalgamated picture of entire universities, and more broadly, entire university systems.

\appendix*

\section{}
This appendix is formatted in numbered sections which mirror the sections in the main paper.  As such, Section~\ref{sec:form} is concerned with the  formulation of the model; here further in-depth discussion concerning the model assumptions is provided.  Section~\ref{sec:theory} contains the proofs of the theorems proposed in the main paper.  Section~\ref{sec:steady} gives the details behind the derivation of the algebraic steady-state solutions.  Section~\ref{sec:dynamics} gives the details behind the derivation of the dynamic solutions in the case where $k=1$.  Section~\ref{sec:decomp} is concerned with the fundamental decomposition of the progression rates, in particular, we prove the theorem that the so-called cascade model implies $k=1$ -- this theorem is only stated (with context) in the main paper.  In Section~\ref{sec:application} we provide details and supporting calculations for the application of our model to real academic promotion systems in Spain, France, and Italy.

\subsection{Formulation}
\label{sec:form}

For completeness, we recall the model assumptions here:
\begin{enumerate}[noitemsep]
\item \label{ass_time} Time is measured in years.
\item \label{ass_homo} Apart from the binarization of the population into the specified groups and career levels, the population is otherwise homogeneous.
\item \label{ass_N} The total organizational headcount grows according to 
\begin{equation}
\frac{\mathd N}{\mathd t}=\lambda N,
\label{eq:app:dNdt}
\end{equation}
where $\lambda$ is a constant with  dimensions of $[\text{Number of indivudals}]^{-1}[\text{Year}]^{-1}$.
\item \label{ass_resign} There are no resignations, deaths in service, redundancies, or dismissals -- employees leave the organization only through retirements.
\item \label{ass_equal} The organization recruits members of both groups at equal rates.  Recruitment is  only at the entry level; access to the managerial level is by promotion only.    Once at the managerial level, employees cannot return to the entry level -- there is no `demotion' of managers.
\item \label{ass_retire} Employees of the $P$- and $Q$-groups retire at equal rates; employees at the different levels in the hierarchy retire at different rates.  
\item \label{ass_rectirex} There is an overall `crude' retirement rate set by the average length of service.
\item \label{ass_fixed} The total number of employees at the managerial level is constrained, with $(P_2+Q_2)/N=\varphi$, where $0<\varphi<1$ is a constant.  Correspondingly, $(P_1+Q_1)/N=1-\varphi$. 
\end{enumerate}
The model represents a first approximation to a university or other public research institution -- these organizations tend to consider growth in headcount to be desirable as (other things being equal) it corresponds to a reduced staff-student ratio and an increase in the inputs required for academic and scientific research.  Elsewhere (in both the private and public sectors), growth in headcount may well be a sign of organizational inefficiency.  This is the justification for emphasizing assumption~\eqref{ass_N}.  Equally, the model assumes low staff turnover, which is typical of universities and public research institutions -- hence assumption~\eqref{ass_resign}.
We emphasize further that the model assumption~\eqref{ass_equal} means that members of the $P$-group and the $Q$-group are recruited into the organization at the entry level in equal numbers.
As such, the model admits an asymmetry between the $P$- and $Q$-groups only by allowing for different progression rates from the entry level to the managerial level for both groups.

It will be helpful to recall the model equations here in a self-contained fashion, for ease of reference in what follows.  We also take the opportunity further to connect the assumptiosn to the model equations.  As such, we recall:
\begin{subequations}
\begin{eqnarray}
\frac{\mathd P_1}{\mathd t}&=&s - r_1 P_1 - \mu P_1,\\
\frac{\mathd P_2}{\mathd t}&=& \mu P_1-r_2 P_2.
\end{eqnarray}%
\label{eq:app:model1}%
\end{subequations}%
The notation is identical to that in the main paper: $s$ is the source function governing the rate at which individuals of the $P$-group are recruited into the organization; this quantity has dimensions of 
$[\text{Number of individuals}][\text{Year}]^{-1}$.  The source function $s$ depends on time.  The other coefficients in Equation~\eqref{eq:app:model1} have the interpretation of rate coefficients, possibly time-dependent.  In particular,
\begin{itemize}
\item $r_1$ is the rate at which members of the $P$-group at the entry level retire.  
\item $\mu$ is the rate at which members of the $P$-group at the entry level are promoted to the managerial level.  
\end{itemize}
These quantities both have dimensions of 
%
%
$[\text{Percentage}][\text{Year}]^{-1}$. 
As such, $r_1$ can be interpreted as the proportion of a $P_1$-individuals who retire per year (`retirement rate'), and $\mu$ can be interpreted as the proportion of all $P_1$-individuals who are promoted to the managerial level, per year (`progression rate').  The fact that $r_1\neq 0$ means that some of the members of the $P$-group at entry level are never promoted to managerial level and spend their whole length of service at the entry level.  Finally, $r_2$ can be interpreted as the proportion of $P_2$-individuals who retire per year (this is again referred to as a `retirement rate').  In general, the rates $r_1$ and $r_2$ will be different.  
It should be noted that Equations~\eqref{eq:app:model1} are valid for $t>0$; at $t=0$, initial conditions apply:
\begin{equation}
P_1(t=0)=P_{10},\qquad P_2(t=0)=P_{20},
\label{eq:app:ic1}
\end{equation}
where $P_{10}$ and $P_{20}$ are constants corresponding to the initial values of the different $P$-populations.
The equations  for the $Q$-group are very similar to those already written down for the $P$-group:
\begin{subequations}
\begin{eqnarray}
\frac{\mathd Q_1}{\mathd t}&=&s - r_1 Q_1 - \mu' Q_1,\\
\frac{\mathd Q_2}{\mathd t}&=& \mu' Q_1-r_2 Q_2.
\end{eqnarray}%
\label{eq:app:model2}%
\end{subequations}%
The source function and the retirement rates are the same in Equations~\eqref{eq:app:model1} and~\eqref{eq:app:model2}, as per Assumptions~\eqref{ass_equal}--\eqref{ass_retire}.  The progression rate in Equation~\eqref{eq:app:model2} is $\mu'$.  Finally, we recall:
\begin{equation}
P_1(t)+P_2(t)+Q_1(t)+Q_2(t)=N(t).
\label{eq:app:N}
\end{equation}

\subsection{Theoretical Analysis}
\label{sec:theory}

In this section we prove the various theorems stated in the main part of the paper, starting with
\begin{theorem}
The source term is not arbitrary; it is given by
\begin{equation}
s(t)=\tfrac{1}{2}\left(\lambda+\hatr\right)N_0\mathe^{\lambda t},
\label{eq:app:source}
\end{equation}
where $\hatr$ is the crude retirement rate.
\end{theorem}
\begin{proof}
We begin by noting that Equations~\eqref{eq:app:model1} and~\eqref{eq:app:model2} may be combined to give
\begin{eqnarray*}
2s-r_1 (P_1+Q_1)-r_2 (P_2+Q_2)&=&\frac{\mathd}{\mathd t}\left(P_1+P_2+Q_1+Q_2\right),\\
                              &\stackrel{\text{Eq.~\eqref{eq:app:N}}}{=}&\frac{\mathd N}{\mathd t},\\
															&\stackrel{\text{Eq.~\eqref{eq:app:dNdt}}}{=}&\lambda N.
\end{eqnarray*}
We therefore have $N=N_0\mathe^{\lambda t}$.    We further identify the crude retirement rate via the definition 
\begin{equation}
\hatr N=r_1 (P_1+Q_1)+r_2 (P_2+Q_2).
\label{eq:app:hatr_def}
\end{equation}
Combining the above equations, we obtain
\begin{equation}
2s-\hatr N=\lambda N,
\label{eq:app:hatr_eq}
\end{equation}
from which the result follows, using $N=N_0\mathe^{\lambda t}$. \myqed
\end{proof}
We also recall:
\begin{theorem}
The retirement rate $r_1$  is given by
\begin{equation}
r_1=\frac{\hatr - r_2\varphi}{1-\varphi}.
\label{eq:app:r1}
\end{equation}
\end{theorem}

\begin{proof}The result~\eqref{eq:app:r1} follows directly from the definition of the crude retirement rate in Equation~\eqref{eq:app:hatr_def}; this equation can be rewritten as
\[
\hatr=r_1\left(\frac{P_1+Q_1}{N}\right)+r_2\left(\frac{P_2+Q_2}{N}\right).
\]
As the ratios $(P_1+Q_1)/N$ and $(P_2+Q_2)/N$ are assumed constant (\textit{cf.} Assumption~\eqref{ass_fixed}), the above can be rewritten as
\[
\hatr=r_1(1-\varphi)+r_2\varphi,
\]
hence $r_1=(\hatr-r_2\varphi)/(1-\varphi)$ and the result is shown. \myqed
\end{proof}

We next revisit the progression rates $\mu$ and $\mu'$.    The starting-point here is to recall that $
(P_2+Q_2)/N=\varphi=\text{Const.}$, hence  $(\mathd/\mathd t)[(P_2+Q_2)/N]=0$, and
\begin{equation}
\frac{\mathd}{\mathd t}(P_2+Q_2)-\frac{1}{N}\frac{\mathd N}{\mathd t}(P_2+Q_2)=0.
\end{equation}
Hence also,
\begin{equation}
\frac{\mathd}{\mathd t}(P_2+Q_2)=\lambda (P_2+Q_2).
\end{equation}
We combine this result with Equations~\eqref{eq:app:model1} and~\eqref{eq:app:model2} to obtain
\begin{equation}
\mu P_1+ \mu' Q_1-r_2 (P_2+Q_2)=\lambda (P_2+Q_2).
\label{eq:app:prog1}
\end{equation}
This motivates the definition of the crude progression rate $\muhat$ in the main paper:
\begin{equation}
\muhat (P_1+Q_1)= \mu P_1+ \mu' Q_1.
\label{eq:app:muhat_def}
\end{equation}
Combining Equations~\eqref{eq:app:prog1} and~\eqref{eq:app:muhat_def} we have the theorem from the main part of the paper:
\begin{theorem}
The crude progression rate $\muhat$ is given by
\begin{equation}
\muhat=(r_2+\lambda)\left(\frac{\varphi}{1-\varphi}\right).
\label{eq:app:r1}
\end{equation}
\end{theorem}
\begin{proof}
We combine Equations~\eqref{eq:app:prog1} and~\eqref{eq:app:muhat_def} to produce
\[
\muhat (P_1+Q_1)=(\lambda+r_2) (P_2+Q_2).
\]
We divide across by $P_1+Q_1$
\[
\muhat =(\lambda+r_2)\left(\frac{P_2+Q_2}{P_1+Q_1}\right).
\]
The ratio $(P_2+Q_2)/(P_1+Q_1)$ can be rewritten as $[(P_2+Q_2)/N]/[(P_1+Q_1)/N]=\varphi/(1-\varphi)$ -- substituting this expression into the above equation gives the required result~\eqref{eq:app:r1}. \myqed
\end{proof}

\subsection{Steady-state solution}
\label{sec:steady}

We recall there the scaled dynamical equations introduced in Section 2 of the main paper:
\begin{subequations}
\begin{eqnarray}
\frac{\mathd p_1}{\mathd t}&=&s_0-(r_1+\lambda) p_1 - \muhat \Nfunc(p_1,q_1)p_1,\\
\frac{\mathd p_2}{\mathd t}&=&\muhat \Nfunc(p_1,q_1)p_1-(r_2+\lambda)p_2,\\
\frac{\mathd q_1}{\mathd t}&=&s_0-(r_1+\lambda) q_1 - \muhat k \Nfunc(p_1,q_1)q_1,\\
\frac{\mathd q_2}{\mathd t}&=&\muhat k \Nfunc(p_1,q_1)q_1-(r_2+\lambda)q_2,
\end{eqnarray}%
\label{eq:app:syst2}%
\end{subequations}%
where $\Nfunc(p_1,q_2)=(p_1+q_1)/(p_1+kq_1)$.
Steady-state solutions of Equation~\eqref{eq:app:syst2} are introduced in Section 3 of the main paper.  These are obtained by setting $\mathd /\mathd t=0$ in Equation~\eqref{eq:app:syst2}.  The purpose of the present section is to showcase the detailed calculations required to derive these steady-state solutions.

By setting $\mathd/\mathd t=0$ in Equation~\eqref{eq:app:syst2}, we obtain the following algebraic equations:
\begin{subequations}
\begin{eqnarray}
s_0=(r_1+\lambda)p_1+\muhat \Nfunc p_1,&\phantom{aaa}&\muhat \Nfunc p_1=(r_2+\lambda)p_2,\\
s_0=(r_1+\lambda)q_1+k\muhat \Nfunc q_1,&\phantom{aaa}&\muhat k\Nfunc q_1=(r_2+\lambda)q_2.
\end{eqnarray}
Combining these equations gives
\begin{equation}
\frac{s_0-(r_2+\lambda)q_2}{s_0-(r_2+\lambda)p_2}=\frac{q_1}{p_1},\qquad k\frac{q_1}{p_1}=\frac{q_2}{p_2}
\label{eq:app:x1}
\end{equation}
\label{eq:app:ss_def}
\end{subequations}
We now work through the different cases of $k$.

\subsection*{Special Case: $k=0$}

The case $k=0$ is anomalous, and corresponds to the steady state $q_2=0$.  Hence, $p_2=\varphi$.  In this case, Equations~\eqref{eq:app:ss_def} reduce to
\begin{eqnarray*}
s_0&=&(r_1+\lambda)p_1+\muhat \Nfunc p_1,\\
s_0&=&(r_1+\lambda)q_1.
\end{eqnarray*}
The second equation here gives
\[
q_1=\frac{s_0}{r_1+\lambda}.
\]
But $p_1+q_1=1-\varphi$, hence
\[
p_1=1-\varphi-\frac{s_0}{r_1+\lambda}.
\]
Since $p_1$ is a population, we require $p_1\geq 0$, hence
\begin{equation}
\varphi\leq \frac{\hatr+\lambda-s_0}{r_2+\lambda}.
\label{eq:app:varphi_cst}
\end{equation}
As $k=0$ is an extreme case, Equation~\eqref{eq:app:varphi_cst} can be thought of as a sufficient condition such that $p_1\geq 0$ at the steady state; indeed, this can be assumed to be a general condition to avoid a population crash where $p_1\rightarrow 0$ in finite time.  As such, in the remainder of this work, we assume that Equation~\eqref{eq:app:varphi_cst} holds.

\subsection*{Special Case: $k=1$}

When $k=1$ the $p$- and $q$-populations are symmetric, and it can be anticipated that $p_2=q_2=\varphi/2$, and $p_1=q_1=(1-\varphi)/2$ in the steady state.  This can be checked by direct calculation on Equations~\eqref{eq:app:ss_def}.

\subsection*{General Solution}


 We introduce $x=q_2/p_2$, hence $x p_2=q_2$.  But $p_2+q_2=\varphi$, hence $p_2=\varphi/(1+x)$.  Hence, Equation~\eqref{eq:app:ss_def} simplifies:
\begin{equation}
\frac{s_0-(r_2+\lambda)\frac{x\varphi}{1+x}}{s_0-(r_2+\lambda)\frac{\varphi}{1+x}}=\frac{x}{k}.
\label{eq:app:x2}
\end{equation}
Equation~\eqref{eq:app:x2} simplifies to a quadratic equation:
\begin{equation}
x^2+(k-1)x\left[\frac{\varphi(r_2+\lambda)}{s_0}-1\right]-k=0,
\label{eq:app:x3}
\end{equation}
with solution
\begin{equation}
x=-\tfrac{1}{2}(k-1)\left[\frac{\varphi(r_2+\lambda)}{s_0}-1\right]\pm
\tfrac{1}{2}\sqrt{(k-1)^2\left[\frac{\varphi(r_2+\lambda)}{s_0}-1\right]^2+4k}.
\label{eq:app:x4}
\end{equation}
%
Since $x=q_2/p_2$ is a population ratio, we choose the positive sign in front of the square-root in Equation~\eqref{eq:app:x4}, which gives $x\geq 0$.

Once $x$ is determined from Equation~\eqref{eq:app:x3}, we work backwards and find $p_2=p_{2*}=\varphi/(1+x)$, $q_2=q_{2*}=x\varphi/(1+x)$ (we herein denote the steady-state solution with an asterisk).  Furthermore, we have $q_1/p_1=x/k$, hence 
\[
p_1=p_{1*}=\frac{1-\varphi}{1+(x/k)},\qquad q_1=q_{1*}=\frac{x}{k}\frac{1-\varphi}{1+(x/k)}.
\]
As such, there is a well-defined steady state of the model wherein $(p_1,q_1,p_2,q_2)$ have definite constant values, fixed by the parameter $x$.  These are summarized here as follows, and match up with the stated result in Section 3 of the main paper:
\begin{equation}
(p_{1*},p_{2*},q_{1*},q_{2*})=
\left(\frac{1-\varphi}{1+(x/k)},
\frac{\varphi}{1+x},
\frac{x}{k}\frac{1-\varphi}{1+(x/k)},
\frac{x\varphi}{1+x}\right).
\label{eq:app:pstar}
\end{equation}

\subsection{Dynamics}
\label{sec:dynamics}

In this section we compute the exact solutions of the model equations~\eqref{eq:app:syst2} which are available when $k=1$.  The methodology, context, and implications are discussed in Section 4 of the main paper, here, we provide the details.  The starting-point is Equation~\eqref{eq:app:syst2} with $\Psi=1$.  Since both populations are now symmetric, it suffices to focus on the $p$-equations, which now read
\begin{subequations}
\begin{eqnarray}
\frac{\mathd p_1}{\mathd t}&=&s_0-(r_1+\lambda) p_1 - \muhat p_1,\label{eq:app:syst3a}\\
\frac{\mathd p_2}{\mathd t}&=&\muhat p_1-(r_2+\lambda)p_2,\label{eq:app:syst3b}
\end{eqnarray}%
\label{eq:app:syst3}%
\end{subequations}%
Equation~\eqref{eq:app:syst3a} is a standard first-order linear ODE, which can be solved via the integrating-factor technique to give
\begin{equation}
p_1(t)=\left[p_1(0)-\frac{s_0}{r_1+\muhat+\lambda}\right]\mathe^{-(r_1+\muhat+\lambda)t}+\frac{s_0}{r_1+\muhat+\lambda}.
\label{eq:app:sln1}
\end{equation}
We substitute this into Equation~\eqref{eq:app:syst3b}, which now reads
\begin{equation}
\frac{\mathd p_2}{\mathd t}=\muhat\bigg\{
\left[p_1(0)-\frac{s_0}{r_1+\muhat+\lambda}\right]\mathe^{-(r_1+\muhat+\lambda)t}+\frac{s_0}{r_1+\muhat+\lambda}
\bigg\}-(r_2+\lambda)p_2.
\label{eq:app:syst3bb}
\end{equation}
Equation~\eqref{eq:app:syst3bb} is a further first-order linear ODE, which has explicit solution
\begin{multline}
p_2(t)=\bigg\{
p_2(0)-\muhat\frac{s_0}{r_1+\muhat+\lambda}\frac{1}{r_2+\lambda}
-\muhat\left[p_1(0)-\frac{1}{r_1+\muhat+\lambda}\right]\frac{1}{(r_2-r_1)-\muhat}\bigg\}\mathe^{-(r_2+\lambda)t}\\
+\muhat\frac{\left[p_1(0)-\frac{1}{r_1+\muhat+\lambda}\right]\mathe^{-(r_1+\muhat+\lambda)t}}{(r_2-r_1)-\muhat}
+\frac{\muhat}{r_2+\lambda}\frac{s_0}{r_1+\muhat+\lambda}.
\label{eq:app:sln2}
\end{multline}
Furthermore, when $k=1$, the $p$-equations and $q$-equations are copies of each other.  As such, the solution of the $q$-equations in this case can be read off from Equations~\eqref{eq:app:sln1} and~\eqref{eq:app:sln2}, with appropriate changes for the initial conditions.

\subsection{Decomposition of the Progression Rate}
\label{sec:decomp}

We recall the following theorem from Section 5 of the main paper:
\begin{theorem}
\label{thm:cascade}
The cascade model requires that $k=1$.
\end{theorem}
In this section we provide the proof, as follows:
\begin{proof}
The number of $Q$s promoted in a given year, expressed as a fraction of all individuals promoted in a given year is
\[
\frac{\mu' Q_1}{\mu'Q_1+\mu P_1}.
\]
The cascade model requires that this should be equal to the number of $Q$s at the entry level, expressed as a fraction of all individuals at the entry level:
\[
\frac{\mu' Q_1}{\mu'Q_1+\mu P_1}=\frac{Q_1}{Q_1+P_1}.
\]
Both sides of this equation can be inverted to give
\[
1+\frac{\mu}{\mu'}P_1=Q_1+P_1.
\]
We divide both sides by $N$ to obtain $1+(\mu/\mu')p_1=q_1+p_1$.  But $q_1+p_1=1-\varphi$ is fixed, hence $(\mu/\mu')p_1=p_1$.  Assuming $p_1\neq 0$, this gives $\mu/\mu'=1$, hence $k=1$. \myqed
\end{proof}

\subsection{Application to Gender Balance in European Universities}
\label{sec:application}

In this section we provide the detailed supporting calculations which underlie the application of the differential-equation model to real promotion systems in Spain, France, and Italy in Section 6 of the main paper.  In what follows, it is helpful to recall the features of the promotion systems in these countries, as well as providing the details of the calculations.

\subsection*{Spain}

\begin{framed}
\noindent Key points:
\begin{itemize}

\item A pioneering paper by Zinovyeva and Bagues~\cite{zinovyeva2010does} examines the Spanish academic promotion system  in the years 2002-2006.
\item In these years, the system was based a national examination (\textit{Habilitaci\'on}):
\begin{itemize}
\item Expert evaluation committees (7 members) are convenened, by random selection from a national  pool of professors.
\item Evaluation is based on a \textit{resum\'e} and research proposal only (Full Professor -- Level A), or on a
\textit{resum\'e}, research proposal, and lectures (Associate Professor -- Level B). 
\item Candidates who qualify in the \textit{habilitaci\'on} may apply for positions at the university level.
\end{itemize} 
\item A regression model is used to determine how the probability $p$ of success in the competition depends on the gender of applicants, as well as other applicant attributes (age, academic productivity, etc.).
\item By random selection, some evaluation committees  have an all-male composition -- this facilitates a `natural experiment' whereby the effect of the committee composition on promotion prospects can be studied.
\item For academics in the competition,  $p$  depends on gender -- males have on average a higher probability of success.  The difference is $p=10.2\%$ for males and $p=8.9\%$ for females.  This is found to be statistically significant.
\item The difference goes away when candidates are assessed by mixed-gender panels.
\item The regression analysis can be used to estimate the parameters of our own model, $\sigma$, $\nu$, etc.
\end{itemize}
\end{framed}

The model makes $p_{ie}$ a linear function of gender, committee composition (number of female members in the evaluation committee), and the number of positions in a particular subject area destined to be filled.  There is a direct effect of gender, and it is found to be statistically significant at a 5\% level -- the average probability of a male candidate succeeding to the full-professor level is 10.2\% whereas the average probability for a female is 8.9\%, a difference of 1.3 percentage points.  

Using the above findings, and the summary statistics included in the study, it is possible to identify parameter values which correspond to $\sigma$, $\nu$, etc. in the four-level model.  We use the given data that in the period reported, males made up $65\%$ of the headcount at the Associate Professor level, and females made up $35\%$ -- these combined make up the pool of applicants for the Full-Professor level.  At the same time, there were 6,037 individuals who applied for promotion to the top level.  However, many applicants applied more than once -- the average number of applications per candidate in the given time period was just over two.  As such, there were 9,480 applications by male candidates, and 3,744 applications by female candidates, corresponding to a total of 13,224 applications.  From these numbers, the ratio $\nu/\nu'$ may be derived:
\begin{multline}
\frac{\nu}{\nu'}=\frac{\left[\begin{array}{c}\text{Number of Males under consideration for promotion to Level A,}\\
\text{as a proportion of all males at Level B}\end{array}\right]\big/\text{[Year]}}
{\left[\begin{array}{c}\text{Number of Females under consideration for promotion to Level A,}\\
\text{as a proportion of all females at Level B}\end{array}\right]\big/\text{[Year]}}
\\
=\frac{9480/0.65}{3744/0.35}\approx 1.36.
\end{multline}
(the fact that the denominator and numerator both appear as rates on a per-year basis means that the numerator and denominator can both be rescaled to appear as rates on a per-five-year basis, to coincide with the duration of the study).
Equally,  $\sigma$ and $\sigma'$ can be identified with the average probability of success for male and female candidates respectively, hence
\begin{equation}
\frac{\sigma}{\sigma'}=\frac{10.2\%}{8.9\%}=1.14.
\end{equation}
Using the difference operator $\Delta\Phi=\Phi_{P}-\Phi_{Q}$, where $\Phi$ is any one of $\mu$, $\nu$, and $\sigma$, we obtain the following relation via the relation
\begin{equation}
\frac{\Delta (\sigma\nu)}{\sigma \nu}=\frac{\Delta \nu}{\nu}+\frac{\Delta \sigma}{\sigma}.
\end{equation}
As such, we can identify $\Delta \nu/\nu=(1.36-1)/1.36=0.26$ and $\Delta \sigma/\sigma=(1.14-1)/1.14=0.12$. 

\subsection*{France}

\begin{framed}
\noindent Key points:
\begin{itemize}
\item A paper by Bosquet et al.~\cite{bosquet2014gender} examines the French academic promotion system  in Economics between 1991-2008.
\item In France, promotion (in all subjects) is based on a national competition (\textit{concours}).  Candidates are evaluated by an evaluation committee.  
\item There are two academic career tracks: the universities, and the research institutes (CNRS).  The \textit{concours} for each career track has its own characteristics.   The paper compares the outcomes of the two types of \textit{concours}. 
\item The data presented in the study consists of academics who applied for promotion, and those who did not.
The study therefore distinguishes between these two groups, and introduces a probability $p(S)$ of success for a candidate, conditional on his/her having applied for promotion.  
\item A regression model for $p(S)$ is constructed.  This shows there are differences in $p(S)$ for males and females, but they are  not  statistically significant.
\item The main reason for the gender asymmetry at the top of the hierarchy is in the difference in the proportion of males and females who enter the promotion competitions, i.e. a supply-side effect.
\item The regression analysis can be used to compute the parameters of our own model, $\sigma$, $\nu$, etc.
\end{itemize}
\end{framed}

 The authors  constructed a linear model for the conditional probability $p_{it}(S)$ -- this is the probability that individual $i$ is promoted at time $t$, conditional on his/her having entered the \textit{concours}.    The probability is a linear function of gender, age, and academic productivity.
The authors find that there is a difference between the average probability of success for males and females, but this is not statistically significant, either for the CNRS or the Universities.  The authors caution however that the results are not fully conclusive, as the lack of significance could be due
to the sample being small (there are, respectively, only 188 female candidates in the University sample and 41 female candidates in the CNRS sample). 

The authors further investigate the probability $p_{it}(A)$ that an individual $i$ from the entire pool of academic economists in the education system will enter the \textit{concours} in a given year $t$.  The model is again a linear function of gender and other variables.  Here, there is is a gender difference, and it is found to be statistically significant at the 1\% level.  As such, the average probabilities $p_{it}(A)$ are $9\%$ per year for males, $6\%$ for females (universities) and $20.3\%$ per year for males, $11.5\%/\text{[Year]}$ for females (CNRS).

These figures are averages of probabilistic variables, and they can be related to the parameters of our own deterministic model ($\nu$, $\sigma$, etc.). 
The authors identify probabilistic variables which can be related to our own deterministic variables $\mu$, $\nu$, and $\sigma$ (our deterministic variables can be thought of as averages of the corresponding probabilistic ones).  As such, the probabilities $p_{it}(A)$ that an individual will enter a \textit{concours} in a given year $t$, averaged over all males or over all females, correspond to $\nu$ and $\nu'$ in our model.
Using this identification, we obtain $\nu/\nu'\approx 1.5$ and $\nu/\nu'\approx 1.76$ for the universities and CNRS respectively.   
Furthermore, that the average probabilities of success conditional on candidates having entered the \textit{concours} are the same for both males and females (have no statistically significant difference) means that  $\sigma=\sigma'$ in our model. 

\newpage
\subsection*{Italy -- Local Competitions}

\begin{framed}
\noindent Key points:
\begin{itemize}
\item  A paper by de Paola and Scoppa~\cite{de2015gender} examines the Italian academic promotion system for Associate and Full Professor positions in the years 2008-2011.
\item In these years, the system was based on competitions organized at a local level:
\begin{itemize}
\item An institution with a vacancy convenes a committee.
\item The committee selects two successful candidates.
\item The top-ranked candidate receives the position within the institution.
\item The second-ranked candidate is deemed promotable and enters a pool from which other institutions can recruit him/her, should a vacancy arise.
\item The rules for the composition of the panel were changed at the end of 2008 -- before, the committee was local, thereafter, only the first member was local.  The remaining 4 committee members were chosen at random from all full professors in the country, in the relevant field.  
\end{itemize} 
\item Study focuses on the competitions for promotion in Economics and Chemistry.  Competitions for promotions to Associate Professor (Level B) and Full Professor (Level A) are grouped together.  Data are obtained from official sources.
\item A regression model is used to determine how the probability $p$ of success in the competitions depends on gender, as well as other factors.
\item For academics in the competitions,  $p$  depends on gender -- males have on average a higher probability of success.  The difference is $p=13\%$ for males and $p=8.9\%$ for females.  This is found to be statistically significant.
\item The difference almost goes away when candidates are assessed by mixed-gender panels.
\item The results from the regression analysis and calculations to disaggregate the data for the Full Professorial level can be combined to generate estimates the parameters of our own model, $\sigma$, $\nu$, etc.
\end{itemize}
\end{framed}
As in the previous works, the authors carry out a regression analysis to model the probability of a candidate's success as a function of gender, age, productivity, etc; the analysis is aggregated over all competitions, for both Associate and Full Professor grades.  Overall, it is found that gender is the key variable that leads to a statistically significant difference in the probability of a candidate's being promoted.  The average probability of success in a competition is 13\% for males and 8.9\% for females, this is found to be statistically significant at the 1\% level.  As in the Spanish study, this difference vanishes for the case of mixed-gender evaluation panels.

We  estimate average success rates and progression rates for Italy, for promotion to the level of Full Professor.  Because the statistical analysis has aggregated the data for the Full Professor and Associate Professor competitions, the following calculations involve some estimates.  

For these purposes, we use the stated fact in the paper that in the period reported, males made up $66\%$ of the headcount at the Associate Professor level, and females made up $34\%$ -- these combined make up the pool of applicants for the full-professor level.  At the same time, the proportion of male candidates for promotion to the top level was $67.2\%$ and the proportion of female candidates was $32.8\%$.  From these numbers, the ratio $\nu/\nu'$ may be derived:
\begin{multline}
\frac{\nu}{\nu'}=\frac{\left[\begin{array}{c}\text{Number of Males under consideration for promotion to Level A,}\\
\text{as a proportion of all males at Level B}\end{array}\right]\big/\text{[Year]}}
{\left[\begin{array}{c}\text{Number of Females under consideration for promotion to Level A,}\\
\text{as a proportion of all females at Level B}\end{array}\right]\big/\text{[Year]}}
\\
=\frac{67.2/66}{32.8/34}\approx 1.06.
\end{multline}
A priori, it is not obvious if this difference is statistically significant.
However, the results of a follow-on study (see below) suggest that the difference is statistically significant.

At the same time, the mean success rate (averaged over male and female candidates) for promotion to the top level is 10.2\%.  There is also a known discrepancy between the success rates for males and females (averaged over all competitions, for both associate and full professor levels), with $\Delta \sigma=4.7\%$, in favour of males.  Assuming rather conservatively that this difference applies uniformly 
across both competitions, we obtain the following simultaneous equations:
\begin{eqnarray*}
\sigma_{\text{avg}}=0.102&=&\sigma_{\text{male}}(0.672)+\sigma_{\text{female}}(0.328),\\
\Delta \sigma=0.047&=&\sigma_{\text{male}}-\sigma_{\text{female}},
\end{eqnarray*}
hence $\sigma_{\text{male}}=\sigma_{\text{avg}}+\Delta\sigma(0.328)=11.7\%$, hence (in the original notation), $\sigma/\sigma'=11.7/(11.7-4.7)=1.67$.
Accordingly, we can estimate
\begin{equation}
\frac{\Delta (\sigma\nu)}{\sigma \mu}=\frac{\Delta \nu}{\nu}+\frac{\Delta \sigma}{\sigma},
\end{equation}
with $\Delta \nu/\nu=(1.06-1)/1.06=0.06$ and $\Delta \sigma/\sigma=(1.67-1)/1.67=0.40$.  Hence, the asymmetry in the progression rates in the Italian data in the years 2008-2011 is due almost entirely to in-competition bias.  

\newpage
\subsection*{Italy -- National Competitions}

\begin{framed}
\noindent Key points:
\begin{itemize}
\item  A follow-on paper by de Paola et al.~\cite{de2015gender_b} examines the reformed Italian academic promotion system for Associate and Full Professor positions in the years 2012-2014.
\item In these years, the system was based a national examination (\textit{Abilitazione Scientifica Nazionale}, ASN):
\begin{itemize}
\item Expert evaluation committees (5 members) are convened, by random selection of all full professors.
\item A candidate submits his/her CV.  The CV is assessed by one of the committees, appropriate by discipline.
\item A candidate who is successful in the ASN is deemed promotable.
\item A university with a vacancy may recruit only candidates who have passed the ASN -- passing the ASN is therefore necessary but not sufficient for promotion.
\end{itemize} 
\item Study focuses on the database of all Italian academics -- those who enter the national competition, and those who do not.
\item The study can therefore identify the probability that an individual will enter the competition.
\item This probability depends on gender --  a female has a lower probability of applying for promotion of 5.2 percentage points.
\end{itemize} 
\end{framed}
In this follow-on paper, a more recent competition (2012-2014) is studied in depth (competitions for both Associate Professor and Full Professor are again grouped together).  Here, data for the entire set of Italian academics is examined (i.e. data corresponding to academics who do not enter the competitions, as well as those that do).  In this way, the authors are able to examine the probability that an individual academic will enter the competition, as well as the probability that an individual academic will succeed in the competition, \textit{conditional on having entered}.  The authors find that overall, a female academic has a lower probability of entering the competition (A difference of 5.2 percentage points). 
The main effect documented in the first study~\cite{de2015gender_b} (different male/female success rates) vanishes for the ASN -- the probability that an individual academic will succeed in the competition, conditional on having entered, is the same for males and females (no statistically significant difference).


\end{document}